\documentclass[aps,showpacs,twocolumn,longbibliography,superscriptaddress]{revtex4-1}
\pdfoutput=1
\usepackage[utf8]{inputenc}
\usepackage[T1]{fontenc}
\usepackage{amsmath}
\usepackage{xcolor}
\colorlet{myPurple}{blue!40!red}
\colorlet{myPurplee}{blue!10!red}
\colorlet{myCyan}{cyan!60!gray}
\colorlet{myRed}{red!66!black}
\usepackage{tikz}
\usepackage{pgfplots}
\pgfplotsset{compat=1.14}
\usepackage[colorlinks=true,citecolor=myRed,urlcolor=myRed,linkcolor=myRed]{hyperref}
\usepackage[normalem]{ulem}
\usepackage{exscale}
\usepackage{bbm}
\usepackage{graphicx}
\usepackage{amsmath}
\usepackage{latexsym}
\usepackage{amsfonts}
\usepackage{amssymb}
\usepackage{times}
\usepackage[T1]{fontenc}
\usepackage{amsthm}
\usepackage{enumerate}
\usepackage{bbold}
\usepackage{color}
\usepackage{nicefrac}
\newcommand{\sket}[1]{{\ensuremath{\lvert#1\rangle}}}
\newcommand{\lket}[1]{{\ensuremath{\left\lvert#1\right\rangle}}}
\newcommand{\ket}[1]{\if@display\lket{#1}\else\sket{#1}\fi}

\newcommand{\sbra}[1]{{\ensuremath{\langle#1\rvert}}}
\newcommand{\lbra}[1]{{\ensuremath{\left\langle#1\right\rvert}}}
\newcommand{\bra}[1]{\if@display\lbra{#1}\else\sbra{#1}\fi}

\newcommand{\sbraket}[2]{{\ensuremath{\langle#1\rvert#2\rangle}}}
\newcommand{\lbraket}[2]{{\ensuremath{\left\langle#1\!\left\rvert\vphantom{#1}#2\right.\!\right\rangle}}}
\newcommand{\braket}[2]{\if@display\lbraket{#1}{#2}\else\sbraket{#1}{#2}\fi}

\newcommand{\sketbra}[2]{{\ensuremath{\lvert #1\rangle\!\langle #2\rvert}}}
\newcommand{\lketbra}[2]{{\ensuremath{\left\lvert #1\right\rangle\!\!\left\langle #2\right\rvert}}}
\newcommand{\ketbra}[2]{\if@display\lketbra{#1}{#2}\else\sketbra{#1}{#2}\fi}
\usepackage{tcolorbox}
\usepackage{mathtools}

\newcommand{\proj}[1]{\ketbra{#1}{#1}}

\newcommand{\junk}{\ket{\xi}}

\newcommand{\tr}{\textrm{Tr}}

\newcommand{\idd}{\mathds{1}}

\newcommand{\M}{\mathsf{M}}
\newcommand{\K}{\mathsf{K}}

\newcommand{\X}{\mathsf{X}}

\newcommand{\Z}{\mathsf{Z}}
\newcommand{\E}{\mathsf{E}}
\newcommand{\D}{\mathsf{D}}
\newcommand{\Pp}{\mathsf{P}}
\newcommand{\Ss}{\mathsf{S}}
\newcommand{\R}{\mathsf{R}}

\usepackage{tikz}
\usepackage{lipsum}
\theoremstyle{plain}
\newtheorem{thm}{Theorem}
\newtheorem{lem}{Lemma}

\usepackage{subcaption}

\usepackage{graphicx}
\usepackage{bm}
\usepackage{dsfont}
\usepackage{tikz}
\usepackage[T1]{fontenc}
\usepackage{amsthm}
\usepackage{array}
\usepackage{amssymb}
\usepackage{amsfonts}
\usepackage{cancel}
\usepackage[toc,page]{appendix}
\usepackage{multirow}
\usepackage{color}
\usepackage{calrsfs}
\usetikzlibrary{backgrounds,decorations.pathreplacing,calc}

\usepackage{tkz-euclide}

\DeclareMathAlphabet{\mathcal}{OMS}{cmsy}{m}{n}

\begin{document}

\title{Self-testing nonlocality without entanglement}

\author{Ivan \v{S}upi\'{c}}
\affiliation{CNRS, LIP6, Sorbonne Universit\'{e}, 4 place Jussieu, 75005 Paris, France}
\email{ivan.supic@lip6.fr}
\author{Nicolas Brunner}
\affiliation{Département de Physique Appliquée, Université de Genève, 1211 Genève, Switzerland}

\date{\today}

\begin{abstract}
Quantum theory allows for nonlocality without entanglement. Notably, there exist bipartite quantum measurements consisting of only product eigenstates, yet they cannot be implemented via local quantum operations and classical communication. In the present work, we show that a measurement exhibiting nonlocality without entanglement can be certified in a device-independent manner. Specifically, we consider a simple quantum network and construct a self-testing procedure. This result also demonstrates that genuine network quantum nonlocality can be obtained using only non-entangled measurements. From a more general perspective, our work establishes a connection between the effect of nonlocality without entanglement and the area of Bell nonlocality.
\end{abstract}

\maketitle

\section{Introduction}

The (un)ability to distinguish certain quantum states via measurements is a central aspect of quantum theory, and central to applications such as quantum key distribution~\cite{BB84,Ekert} and data-hiding~\cite{QDH}. This question is of particular interest when considering composite systems. One may consider for instance two remote observers, Alice and Bob, sharing a state chosen from a certain set. Alice and Bob should try to identify which state they received. In this context, a natural limitation is that Alice and Bob must restrict to local measurements (or operations) assisted by classical communication (LOCC). Surprisingly in this case, there exist sets of product states forming a basis, which nevertheless cannot be perfectly distinguished by any LOCC measurement. While Ref.~\cite{PeresWootters} provided preliminary results in this direction, the first examples were constructed by Bennett et al.~\cite{Bennett1999}, coining the effect ``nonlocality without entanglement''. This shows that separable quantum measurements (where all eigenstates are non-entangled) are strictly stronger than LOCC measurements. These ideas have been generalized in many different directions see e.g.~\cite{UPB,Walgate2002,Niset,Cohen,FengShi,Childs,Croke}, with connections to the notion of unextendible product basis and bound entanglement~\cite{UPB,DiVincenzo}. In recent years, renewed interest has been devoted to these ideas, with the discovery of stronger forms of this effect in particular in multipartite systems, see e.g.~\cite{Halder2019,Banik2,Ha}.

In this work, we discuss the question of certifying the effect of ``nonlocality without entanglement'' (NLWE) in a black-box setting. Specifically, we consider the quantum measurement featuring NLWE introduced in~\cite{Bennett1999} (for a two-qutrit system), and show that it can be certified in a device-independent manner. For this we consider the simple quantum network of entanglement swapping~\cite{Zukowski1993} (also known as ``bilocality'' network~\cite{Branciard_2010}), where the middle party performs the NLWE measurement (see Fig.~\ref{fig:domino}). Based on the assumption that the two quantum sources present in the network are independent, the standard assumption in network nonlocality~\cite{Branciard_2010,Fritz_2012}, we can show that NLWE can be certified using concepts and tools from self-testing~\cite{MayersYao,SupicBowles}, a framework for the device-independent certification of quantum resources. In fact, the full quantum setup can be self-tested, including also the shared entangled states and the local measurements of the side parties. Finally, we discuss how our self-testing scheme can be generalized. 

Our result shows that a strong form of nonlocal quantum correlations in networks, known as genuine network quantum nonlocality~\cite{Supic2022}, can be obtained without the need for entangled measurements (as traditionally used in network nonlocality). From a more general point of view, our work also finally connects the effect of ``nonlocality without entanglement'' with the area of Bell nonlocality~\cite{Bell,review}.

\section{Problem}

\begin{figure}
\centering
    \begin{subfigure}[b]{1\columnwidth}
         \centering
	\includegraphics[width = 1\columnwidth]{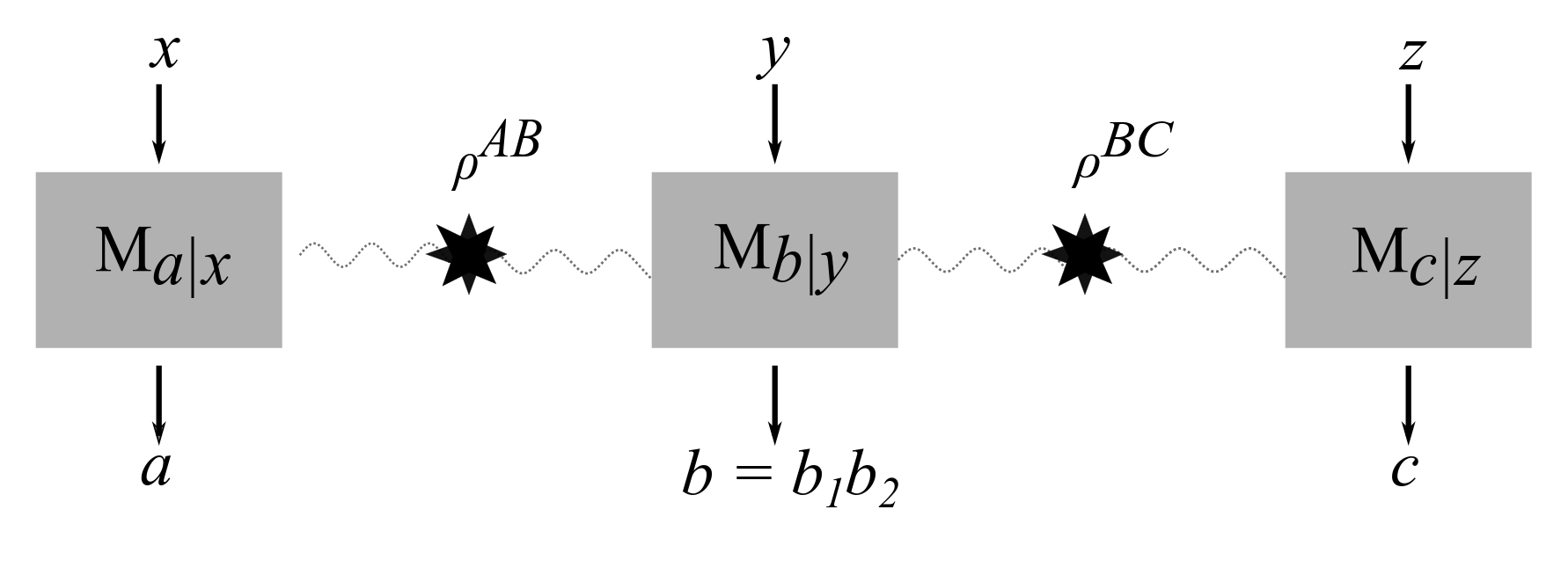}
	\caption{}
	\label{fig:scenario}
	\end{subfigure}
	\hfill
	\begin{subfigure}[b]{1\columnwidth}
         \centering
	\includegraphics[width = 0.6\columnwidth]{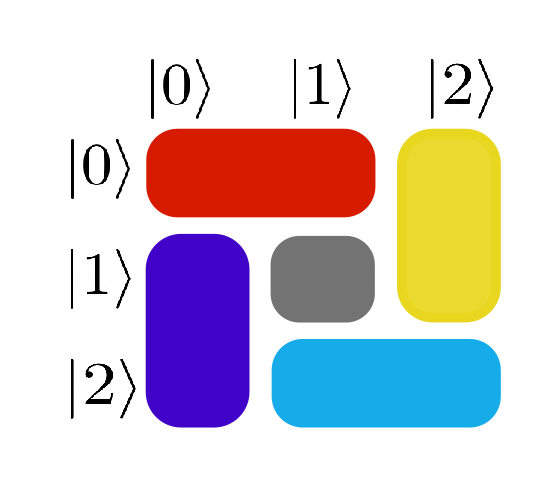}
	\caption{}
	\label{fig:domino}
	\end{subfigure}
	\caption{We consider the entanglement swapping (or bilocality) network represented in (a), where the middle party Bob performs measurements on two incoming subsystems. Our goal is to self-test one of these measurements, denoted $\M'_{\lozenge}$, which features nonlocality without entanglement. That is, the measurement consists of only product eigenstates---as shown by the domino tiles in (b)---yet it cannot be implemented by LOCC. The red tile corresponds to the two projectors $\M'_{0,1|\lozenge}$ and $\M'_{0,2|\lozenge}$ from Eq. \eqref{dominomeas}, and so on.
	}
\end{figure}

We consider the entanglement swapping experiment~\cite{Zukowski1993} consisting of two sources and three parties (nodes), as shown in Fig.~\ref{fig:scenario}. We call the central party Bob, and two lateral ones Alice and Charlie. The center of out interest are the correlations obtained in this setting, the so-called bilocality network~\cite{Branciard_2010,branciard2012bilocal}, where the two sources are assumed to be independent from each other. Note that the characterisation of local and quantum correlations in networks featuring independent sources has attracted growing attention in recent years, see e.g~ \cite{Branciard_2010,Fritz_2012,gisin2020constraints,Renou_2019, Pozas_Kerstjens_2019} and~\cite{ArminReview} for a review.

In this work, our main goal is to certify in a device-independent manner that Bob performs a quantum measurement featuring NLWE. That is, assuming only the independence of the two sources, we will show that, from observed data alone, the presence of such a measurement can be demonstrated, up to irrelevant local transformations. We use the tools and concepts of self-testing~\cite{MayersYao,SupicBowles}, in particular the results of Ref.~\cite{Yang}. While self-tests have been developed for specific joint entangled measurements (such as the well-known Bell-state measurement)~\cite{renou2018self,Bancal_BSM}, we show here that a similar construction is possible for relevant measurements with only separable eigenstates.

Self-testing is a procedure which establishes a form of equivalence between two experiments. First, we have the physical experiment, which corresponds to the actual experiment performed in the laboratory, featuring a priori unknown states and measurements to be certified. The second experiment is termed the reference (or idealized) experiment, or the target for the certification. We say that the physical experiment simulates the reference experiment if it reproduces exactly its statistics. If such a simulation is enough to infer the existence of a local isometry mapping the physical experiment to the reference experiment, we say that the reference experiment is self-tested. In other words, this shows that the physical experiment must be equivalent (up to irrelevant local transformations) to the reference experiment.

We start by describing the reference experiment. The left source, connecting Alice and Bob, prepares a pair of maximally entangled qutrits. The right source, connecting Bob and Charlie, prepares the same state. Hence, the reference states are given by
\begin{align}
    \ket{\psi'}^{A'B_1'} = \ket{\psi'}^{B_2'C'} = \ket{\phi_+} 
\end{align}
where $\ket{\phi_+} = (\ket{00}+\ket{11}+\ket{22})/\sqrt{3}$. Note that we use superscripts to identify the various systems, and that we distinguish the two subsystems of Bob by $B_1'$ and $B_2'$. The reference measurements for Alice are given by three ternary measurements:

\begin{align*}
    \M'_{0} &= \left\{\M'_{0|0} = \proj{0}, \M'_{1|0} = \proj{1}, \M'_{2|0} = \proj{2} \right\},\\
    \M'_{1} &= \left\{\M'_{0|1} = \proj{+}_{0,1}, \M'_{1|1} = \proj{-}_{0,1}, \M'_{2|1} = \proj{2} \right\},\\
    \M'_{2} &= \left\{\M'_{0|2} = \proj{0}, \M'_{1|2} = \proj{+}_{1,2}, \M'_{2|2} = \proj{-}_{1,2} \right\},
\end{align*}
where $\proj{\pm}_{j,k} \equiv (\ket{j}\pm \ket{k})(\bra{j}\pm \bra{k})/{2}$.
Here $\M'_{x}$ denotes the projective measurement for input $x=0,1,2$. Each measurement produces a ternary output $a$, with POVM elements denoted $\M'_{a|x}$. The reference measurements for Charlie are the same as for Alice, and we will denote them by $\M'_{z}$ with elements $\M'_{c|z}$. 

The reference measurements of Bob are of two types. First, we have four measurements with a clear subsystem separation. On the subsystem $B_1$, these are used to self-test the state shared with Alice, as well as Alice's measurements. On the subsystem $B_2$, these are used to self-test the state shared with Charlie, and Charlie's measurements. More precisely, these measurements take the form $\M'_{b_1,b_2|y} = \M'_{b_1|y}\otimes \M'_{b_2|y}$ for $b_1,b_2=0,1,2$ and $y = 0,1,2,3$, where the POVM elements $\M'_{b_1|y}$ and $\M'_{b_2|y}$ correspond to the eigenvectors of the following four operators:
\begin{align*}
    \frac{1}{\sqrt{2}}\begin{bmatrix}
    1 & \pm1 & 0\\ \pm1 & \mp1 & 0\\ 0& 0 & \sqrt{2}
    \end{bmatrix}, \qquad \frac{1}{\sqrt{2}}\begin{bmatrix}
    \sqrt{2} & 0 & 0\\ 0 & 1 & \pm 1\\ 0& \pm 1 & \mp 1
    \end{bmatrix},
\end{align*}
The second type of measurement for Bob is our main object of interest. This corresponds to a single extra measurement, which corresponds to the input $y=\lozenge$. This measurement exhibits the property of nonlocality without entanglement, and is denoted  ${\mathcal{M}'_{\lozenge}} = \{\M'_{b_1,b_2|\lozenge}\}_{b_1,b_2=0}^2$.
It consists of nine eigenstates, which are all product with respect to the partition $B_1$ vs $B_2$, given by
\begin{widetext}
\begin{align} \nonumber
    \M'_{0,0|\lozenge} = \proj{1}\otimes\proj{1}, \qquad \M'_{0,1|\lozenge} = \proj{0}\otimes\proj{+}_{0,1}, \qquad \M'_{0,2|\lozenge} = \proj{0}\otimes\proj{-}_{0,1},\\ \label{dominomeas}
    \M'_{1,0|\lozenge} = \proj{2}\otimes\proj{+}_{1,2}, \qquad \M'_{1,1|\lozenge} = \proj{2}\otimes\proj{-}_{1,2}, \qquad \M'_{1,2|\lozenge} = \proj{+}_{1,2}\otimes\proj{0},\\ \nonumber
    \M'_{2,0|\lozenge} = \proj{-}_{1,2}\otimes\proj{0}, \qquad \M'_{2,1|\lozenge} = \proj{+}_{0,1}\otimes\proj{2}, \qquad \M'_{2,2|\lozenge} = \proj{-}_{0,1}\otimes\proj{2}.
\end{align}
\end{widetext}
The key property of this set of states, is that they cannot be perfectly distinguished via local measurements assisted with classical communication (so-called LOCC measurements) \cite{Bennett1999}. Therefore, the measurement $\mathcal{M}'_{\lozenge}$ cannot be implemented via LOCC, illustrating the fact that separable measurements are strictly stronger than LOCC measurements.

Combining the above states and measurements, we obtain the statistics of the reference experiment, given by the joint conditional probability distribution
\begin{multline} \label{ReferenceCorrelations}
    p'(a,b_1,b_2,c|x,y,z) =  \\ = \tr\left[\left({\M'}_{a|x}^{A'}\otimes {\M'}_{b_1,b_2|y}^{B_1'B_2'}\otimes {\M'}_{a|x}^{C'} \right){\psi}_+^{A'B'_1} \otimes {\psi}_+^{B'_2C'}\right].
\end{multline}
where $x,z=0,1,2$, $y=0,1,2,3,\lozenge$ and $a,b_1,b_2,c=0,1,2$. Note that when computing the above equation, one should be careful about the order of the subsystems.

\section{Main result}

We now present our main result, namely that the reference experiment is a self-test. Consider a physical experiment with a priori unknown states $\ket{\psi}^{AB_1}$ and $\ket{\psi}^{B_2C}$ and measurements $\{\M_{a|x}\}$, $\{\M_{b_1,b_2|y}\}$ and $\{\M_{c|z}\}$, resulting in observed correlations
\begin{multline} \label{PhysicalCorrelations}
    p(a,b_1,b_2,c|x,y,z) =  \\ = \tr\left[\left({\M}_{a|x}^{A}\otimes {\M}_{b_1,b_2|y}^{B_1B_2}\otimes {\M}_{a|x}^{C} \right){\psi}^{AB_1} \otimes {\psi}^{B_2C}\right].
\end{multline}
Below we show that if these statistics correspond to those of the reference experiment, as given in Eq.~\eqref{ReferenceCorrelations}, then all states and measurements of the physical experiment are equivalent (up to irrelevant local transformation) to the reference states and measurements. In particular, this implies that the measurement $y=\lozenge$ for Bob must feature NLWE. More formally, we have the following theorem.

\begin{thm}\label{theorem} Consider a physical experiment such that its statistics, as given in Eq. \eqref{PhysicalCorrelations}, match exactly the statistics of the reference experiment given in Eq. \eqref{ReferenceCorrelations}. Then there exists a local isometry $\tilde{\Phi}$ mapping
\begin{itemize}
    \item Bob's measurement $\M_\lozenge$ to the NLWE measurement $\mathcal{M}'_{\lozenge}$:
\begin{multline}\label{stdomino}
\tilde{\Phi}\left(\M_{b|\lozenge}^{B_1B_2}\ket{\psi}^{AB_1}\otimes\ket{00}^{A'B_1'}\otimes\ket{\psi}^{B_2C}\otimes\ket{00}^{B_2'C'}\right) = \\
     = \ket{\tilde{\xi}}^{AB_1B_2C}\otimes{\M'}_{b|\lozenge}^{B_1'B_2'}\ket{\phi_+}^{A'B_1'}\otimes\ket{\phi_+}^{B_2'C'},
\end{multline}
where $\ket{\tilde{\xi}}^{AB_1B_2C} = \ket{\xi_1}^{AB_1}\otimes\ket{\xi_2}^{B_2C}$ is a valid quantum state.
    \item  the states $\ket{\psi}^{AB_1}$ and $\ket{\psi}^{B_2C}$ to the reference states (i.e. pairs of maximally entangled qutrits), and Alice's and Charlie's measurements to the corresponding reference measurements $\M'_j$:
    \begin{multline}\label{iso1}
    \tilde{\Phi}\left(\M_{a|x}\ket{\psi}^{A_1B_1}\otimes\M_{c|z}\ket{\psi}^{B_2C}\otimes\ket{0000}^{A'B_1'B_2'C'}\right) = \\ = \ket{\tilde{\xi}}^{AB_1B_2C}\otimes\M'_{a|x}\ket{\phi_+}^{A'B_1'}\otimes\M'_{c|z}\ket{\phi_+}^{B_2'C'}.
\end{multline}
\end{itemize}
\end{thm}
\begin{proof}
The theorem consists of two self-testing results, and we start by proving the second one. The idea is to first prove that in order to reproduce some of the marginals of the reference correlations, specifically $\{p(a,b_1|x,y)\}$, implies the equivalence between $\ket{\psi}^{AB_1}$ and $\ket{\phi_+}$. The formal statement is given in the following lemma.

\begin{table}
    \begin{tabular}{|c|c|}
    \hline
        $\quad$ Input set $(x,z)$ $\quad$ & $\quad$ Matching output set $(a,b_1,b_2,c)$ $\quad$ \\
        \hline
        \hline
        $(0,0)$ & $(1,0,0,1)$ \\
        \hline
        $(0,1)$ & $(0,0,1,0)$\\
        & $(0,0,2,1)$\\
        \hline
         $(0,2)$ & $(2,1,0,1)$\\
        & $(2,1,1,2)$\\
        \hline
          $(2,0)$ & $(1,1,2,0)$\\
        & $(2,2,0,0)$\\
        \hline
          $(1,0)$ & $(0,2,1,2)$\\
        & $(1,2,2,2)$\\
        \hline
    \end{tabular}
    \caption{For Bob's input $y=\lozenge$, certain sets of inputs $(x,z)$ for Alice and Charlie imply  specific outputs patterns.}\label{table}
\end{table}

\begin{lem}\label{lem2}
Let the state $\ket{\psi}^{AB_1}\otimes\ket{\psi}^{B_2C}$ and measurements $\{\M_{a|x}\}$ and $\{\M_{b_1|y}\}$ such that $\M_{b_1|y} = \sum_{b_2}\M_{b_1,b_2|y}$ produce correlations $p(a,b_1|x,y)$. If those correlations are such that
\begin{multline}
p(a,b_1|x,y) = \emph{\tr}\left[\left({\M'}_{a|x}^{A'}\otimes {\M'}_{b_1,|y}^{B'_1}\right)\left({\phi}_+^{A'B'_1}\right)\right],
\end{multline}
for every $a,b_1,x,y$ then there exists isometry $\Phi$ such that
\begin{multline}\label{iso}
    \Phi\left(\M_{a|x}\ket{\psi}^{A_1B_1}\otimes\ket{\psi}^{B_2C}\otimes\ket{00}^{A'B'}\right) = \\ = \junk^{ABC}\otimes\left(\M'_{a|x}\ket{\phi_+}^{A'B_1'}\right),
\end{multline}
where $\junk^{ABC}$ is a valid quantum state.
\end{lem}
The proof can be found in Appendix~\ref{suppmat1} and it is directly inspired by the self-testing of maximally entangled pair of qutrits presented in~\cite{Yang}. Combined with methods for self-testing joint measurements introduced in~\cite{Supic2021} allows us to get the self-testing result for both states, as well as all measurements performed by Alice and Charlie, as stated in eq.~\eqref{iso1}. 
The proof of this equation is given in Appendix~\ref{suppmat2}.

Based on these self-testing results, we move on to the second part of the proof, for certifying the additional measurement of Bob, namely $\M_{\lozenge}$ from Eq.~\eqref{stdomino}. First, note that simulation of the reference correlations involves
\begin{align}\nonumber
    p(b|\lozenge) &= \tr\left[\M_{b|\lozenge}^{B}\left(\psi^{AB_1}\otimes\psi^{B_2C}\right)\right] = \frac{1}{9}  \qquad \forall b.
\end{align}
This implies that the norm of vectors $\M_{b|\lozenge}\left(\ket{\psi}^{AB_1}\otimes\ket{\psi}^{B_2C}\right)$ for all outcomes $b$ equals $1/3$. Similarly, from the simulation of the reference correlations we have:
\begin{equation}
    p(a,c|x,z) = \frac{1}{9} \qquad \forall a,c,x,z,
\end{equation}
implying that the norm of vectors $\M_{a|x}\ket{\psi}^{A_1B_1}\otimes \M_{c|z}\ket{\psi}^{B_2C}$ for all $a,c,x,z$ equals $1/3$.

Let us define for certain input sets $(x,\lozenge,z)$ their matching outputs sets $(a,b,c)$, as given in Table~\ref{table}. The simulation of the reference correlations imposes that for every set of inputs from Table~\ref{table}, the matching set of outputs happens with probability $1/9$. Let us now concentrate on the set of inputs $(0,\lozenge,0)$ and its matching set of outputs $(1,0,0,1)$, \emph{i.e.} eq.~$p(1,0,0,1|0,\lozenge,0) = 1/9$. Given eq. \eqref{iso} we obtain the following set of relations:
\begin{widetext}
\begin{align*}
    p(1,0,0,1|0,\lozenge,0) &=  \bra{\psi}^{AB_1}\otimes\bra{\psi}^{B_2C}\M^A_{1|0}\otimes \M^{B_1B_2}_{0,0|\lozenge}\otimes\M^C_{1|0}\ket{\psi}^{AB_1}\otimes\ket{\psi}^{B_2C}\\
       &=  \left(\bra{\psi}^{AB_1}\otimes\bra{00}^{A'B_1'}\otimes\bra{\psi}^{B_2C}\otimes\bra{00}^{B'_2C'}\M^A_{1|0}\otimes\M^C_{1|0}\otimes\idd^{A'B_1B_2B'_1B'_2C'}\right)\Phi^\dagger\times \\ &\qquad\qquad\times \Phi\left(\M_{0,0|\lozenge}^{B_1B_2}\otimes\idd^{AA'B_1'B_2'CC'}\ket{\psi}^{AB_1}\otimes\ket{00}^{A'B_1'}\otimes\ket{\psi}^{B_2C}\otimes\ket{00}^{B_2'C'}\right)\\
       &= \bra{\xi_1}^{AB_1}\otimes\bra{\xi_2}^{B_2C}\otimes\bra{\phi_+}^{A'B'_1}{\M'}_{1|0}^{A'}\otimes\bra{\phi_+}^{B'_2C'}{\M'}^{C'}_{1|0}\times\\ &\qquad\qquad\times\Phi\left(\M_{0,0|\lozenge}^{B_1B_2}\otimes\idd^{AA'B_1'B_2'CC'}\ket{\psi}^{AB_1}\otimes\ket{00}^{A'B_1'}\otimes\ket{\psi}^{B_2C}\otimes\ket{00}^{B_2'C'}\right) = \frac{1}{9}
\end{align*}
\end{widetext}
The second relation is the consequence of the fact that isometries do not change the scalar product. Given that both vectors in the scalar product in the third equality have norm $1/3$ the saturation of Cauchy-Schwarz inequality implies
\begin{multline}
    \Phi\left(\M_{0,0|\lozenge}^{B_1B_2}\ket{\psi}^{AB_1}\otimes\ket{00}^{A'B_1'}\otimes\ket{\psi}^{B_2C}\otimes\ket{00}^{B_2'C'}\right) = \\
     = \frac{1}{9}\ket{\xi_1}^{AB_1}\otimes\ket{\xi_2}^{B_2C}\otimes\ket{11}^{A'B_1'}\otimes\ket{11}^{B_2'C'}
\end{multline}
With similar argumentation we obtain the following set of relations for all $b = b_1,b_2$:
\begin{align*}
    &\Phi\left(\M_{b|\lozenge}^{B_1B_2}\ket{\psi}^{AB_1}\otimes\ket{00}^{A'B_1'}\otimes\ket{\psi}^{B_2C}\otimes\ket{00}^{B_2'C'}\right) = \\
     & = \frac{1}{9}\ket{\xi_2}^{AB_1}\otimes\ket{\xi_2}^{B_2C}\otimes\ket{\psi'_{b}}^{A'}\otimes\ket{\psi'_{b}}^{B_1'}\otimes\ket{\phi'_{b}}^{B_2'}\otimes\ket{\phi'_{b}}^{C'},
     \end{align*}
where states $\ket{\psi'_{b}}$ and  $\ket{\phi'_{b}}$ are such that $\M'_{b|\lozenge} = \proj{\psi'_{b}}\otimes\proj{\phi'_{b}}$.
All these equations together with eq.~\eqref{finalst} imply the self-testing result we needed as given in eq.~\eqref{stdomino}.
\end{proof}

\section{General construction}

The above construction can be generalized to other measurements featuring NLWE, for higher dimensions (still in the bipartite case) and to the multipartite case. The general idea is that the above procedure allows one to basically self-test any measurement with product rank-one eigenstates. When the later involve only real parameters, the construction is rather straightforward, while the general case with complex coefficients is more challenging, as usual in self-testing~\cite{Supic2021}. 

For bipartite NLWE measurements, the bilocality network can be readily used. 
If the measurement $\M'_{b|\lozenge}$ acts now on $\mathbb{C}^{d_A}\otimes\mathbb{C}^{d_B}$, the local dimensions of two maximally entangled states distributed in the network must be adapted (to $d_A$ and $d_B$ respectively). These states, and the local measurements of Alice and Charlie, can be self-tested using any of the available methods (see for example~\cite{Yang,Supic2021}). In turn, Alice and Charlie can remotely prepare for Bob (via their certified local measurements acting on half of the shared maximally entangled pairs) the pair of input states in order to match any of the eigenstates of the measurement $\M'_{b|\lozenge}$.

Moving to the multipartite case will involve a star network, where the central node will perform the measurement with NLWE. For an N-party measurement with qudits, we consider a star network with $N$ branches. On each branch a maximally entangled state of two qudits must first be self-tested, as well as local measurements of the lateral nodes. Second the central NLWE is self-tested as above.

\section{Discussion}

We discussed the device-independent certification of the effect of ``nonlocality without entanglement. Specifically, we showed that a quantum measurement featuring only separable eigenstates, but which cannot be implemented via an LOCC procedure, can be certified in a quantum network with independent sources, based only on observed statistics. 

A point worth noting is that our self-test construction has interesting consequences from the perspective of network nonlocality. In particular, this example shows that genuine quantum network nonlocality \cite{Supic2022}, a form of quantum nonlocality that can only arise in networks, is in fact possible without involving any entangled measurement. 

Let us point out that previous works have discussed the certification of non-classical quantum measurements in a partially device-independent setting, considering prepared quantum states of limited dimension \cite{Vertesi2011,Bennet2014}, also with an experimental demonstration \cite{Bennet2014}. A key difference with our work (besides the stronger assumptions), is that these previous results could only certify that a measurement is not achievable via LOCC, but could not certify the property of NLWE, as we do here.


From a more general perspective, our work establishes a connections between two forms of nonlocality in quantum theory, namely Bell nonlocality and the effect of nonlocality without entanglement.

\begin{acknowledgments}
The authors acknowledge financial support from the Starting ERC grant QUSCO and the Swiss National Science Foundation (project $2000021\_192244/1$ and NCCR SwissMAP).
\end{acknowledgments}

\onecolumngrid
\appendix

\section{Proof of Lemma~\ref{lem2}}\label{suppmat1}

In this section we prove Lemma~\ref{lem2}. The proof is largely inspired by the proof offered in~\cite{Yang}. The most important part of our contribution is self-testing of measurements as well, and not just the state as in~\cite{Yang}. It will be convenient to define the marginal physical measurements operators for Bob $\M^B_{b_1|y} = \sum_{b_2}\M^B_{b_1,b_2|y}$. Not that this physical coarse-grained measurement can, in principle, act on both Hilbert spaces $\mathcal{H}^{B_1}$ and $\mathcal{H}^{B_2}$. Let us further introduce the following notation:
\begin{align}
 \Z_{0,1}^A = \M^A_{0|0} - \M^A_{1|0}, \qquad \Z_{1,2}^A = \M^A_{1|0} - \M^A_{2|0}, &\qquad \X_{0,1}^A = \M^A_{0|1} - \M^A_{1|1}, \qquad \X_{1,2}^A = \M^A_{1|2} - \M^A_{2|2},\\
 \D_{0,1}^B = \M^B_{0|0} - \M^B_{1|0}, \qquad \D_{1,2}^B = \M^B_{1|2} - \M^B_{2|2}, &\qquad \E_{0,1}^B = \M^B_{0|1} - \M^B_{1|1}, \qquad \E_{1,2}^B = \M^B_{1|3} - \M^B_{2|3},\\ \label{hatoperators}
 {\hat{\Z}}^B_{0,1} = \frac{\D_{0,1}^B + \E_{0,1}^B}{\sqrt{2}}, \qquad {\hat{\X}}^B_{0,1} = \frac{\D_{0,1}^B - \E_{0,1}^B}{\sqrt{2}},&\qquad  {\hat{\Z}}^B_{1,2} = \frac{\D_{1,2}^B + \E_{1,2}^B}{\sqrt{2}}, \qquad {\hat{\X}}^B_{1,2} = \frac{\D_{1,2}^B - \E_{1,2}^B}{\sqrt{2}}.
\end{align}
Physical operators $\Z^A_{i,j}$ and $\X_{i,j}^A$ are supposed to act as Pauli's $\sigma_z$ and $\sigma_x$ respectively on qubit subspaces spanned by vector basis $\{\ket{i}^A,\ket{j}^A\}$, hence the notation. Of course, these operators are uncharacterized and only after self-testing result is proven such notation is justified. Operators  $\hat{\Z}^B_{i,j}$ and $\hat{\X}_{i,j}^B$ are supposed to act as Pauli's $\sigma_z$ and $\sigma_x$ respectively on qubit subspaces spanned by vector basis $\{\ket{i}^B,\ket{j}^B\}$, but for the beginning we note that these operators unlike $\Z^A_{i,j}$ and $\X_{i,j}^A$ do not even need to be unitary. Let us now define operators which in ideal case on designated qubit subspaces act as identities:
\begin{align}
    \idd_{0,1,Z}^A = \M^A_{0|0} + \M^A_{1|0}, \qquad \idd_{1,2,Z}^A = \M^A_{1|0} + \M^A_{2|0}, &\qquad \idd_{0,1,X}^A = \M^A_{0|1} + \M^A_{1|1}, \qquad \idd_{1,2,X}^A = \M^A_{1|2} + \M^A_{2|2},\\
 \idd_{0,1,D}^B = \M^B_{0|0} + \M^B_{1|0}, \qquad \idd_{1,2,D}^B = \M^B_{1|2} + \M^B_{2|2}, &\qquad \idd_{0,1,E}^B = \M^B_{0|1} + \M^B_{1|1}, \qquad \idd_{1,2,E}^B = \M^B_{1|3} + \M^B_{2|3}
\end{align}
The operators $\idd_{0,1,D}^B$, $\idd_{0,1,E}^B$, $\idd_{1,2,D}^B$ and $\idd_{1,2,E}^B$ are projectors to the subspaces spanned by eigenvectors of $\D_{0,1}^B$, $\E_{0,1}^B$, $\D_{1,2}^B$ and  $\E_{1,2}^B$. In a similar fashion we want to define projectors to the subspaces spanned by ${\hat{\Z}}^B_{0,1}$, ${\hat{\X}}^B_{0,1}$, ${\hat{\Z}}^B_{0,1}$ and ${\hat{\X}}^B_{0,1}$. However, as we said earlier those operators are not necessarily unitary. However, we define the regularized version of this operators in the following way:
\begin{equation}
    {\Z}^B_{0,1} = \frac{{\hat{\Z}}^B_{0,1}}{|{\hat{\Z}}^B_{0,1}|}, \qquad {\X}^B_{0,1} = \frac{{\hat{\X}}^B_{0,1}}{|{\hat{\X}}^B_{0,1}|}, \qquad {\Z}^B_{1,2} = \frac{{\hat{\Z}}^B_{1,2}}{|{\hat{\Z}}^B_{1,2}|}, \qquad {\X}^B_{1,2} = \frac{{\hat{\X}}^B_{1,2}}{|{\hat{\X}}^B_{1,2}|}.
\end{equation}
Such renormalization of eigenvalues is not possible if any of the operators figuring in the numerators have eigenvectors with a corresponding eigenvalue equal to zero. If such case appears we just change all such eigenvalues from $0$ to $1$. In this way we obtain unitary operators ${\Z}^B_{i,j}$ and ${\X}^B_{i,j}$. Now we define subspace $\mathcal{B}_{i,j}$ which comprises the range of operators $D_{i,j}^B$ and $E_{i,j}^B$. Let $\idd_{i,j}^B$ be the projector to the subspace $\mathcal{B}_{i,j}$. The definition of  ${{\Z}}^B_{i,j}$ and ${{\X}}^B_{i,j}$ implies that the their range is exactly the subspace $\mathcal{B}_{i,j}$, and since all their eigenvalues are either $1$ or $-1$ the following equations hold
\begin{equation}
    {{\Z}^B_{i,j}}^2 = \idd_{i,j}^B, \qquad {{\X}^B_{i,j}}^2 = \idd_{i,j}^B
\end{equation}
Let us now define the following subnormalized states:
\begin{align}
    \ket{\psi_{i,j,Z,A}} = \idd_{i,j,Z}^A\ket{\psi} &\qquad \ket{\psi_{i,j,X,A}} = \idd_{i,j,X}^A\ket{\psi}\\
     \ket{\psi_{i,j,D,B}} = \idd_{i,j,D}^B\ket{\psi} &\qquad \ket{\psi_{i,j,E,B}} = \idd_{i,j,E}^B\ket{\psi},\qquad 
     \ket{\psi_{i,j,B}} = \idd_{i,j}^B\ket{\psi},
\end{align}
where we used notation $\ket{\psi}\equiv \ket{\psi}^{ABC} = \ket{\psi}^{AB_1}\otimes\ket{\psi}^{B_2C}$. Whenever state is written without any superscript it means that we are taking into account the whole state distributed in the network. Let us now analyse the consequences of the fact that physical experiment simulates the reference one. In the first step we concentrate on the following set of correlations:
\begin{align} \label{jedan}
    \bra{\psi}\M_{2|0}^A\otimes \idd^B\ket{\psi} = \bra{\psi}\M_{2|1}^A\otimes \idd^B\ket{\psi} = \bra{\psi}\idd^A\otimes\M_{2|0}^B\ket{\psi} = \bra{\psi}\idd^A\otimes\M_{2|1}^B\ket{\psi} = \frac{1}{3}\\ \label{dva}
    \bra{\psi}\M_{0|0}^A\otimes \idd^B\ket{\psi} = \bra{\psi}\M_{0|2}^A\otimes \idd^B\ket{\psi} = \bra{\psi}\idd^A\otimes\M_{0|2}^B\ket{\psi} = \bra{\psi}\idd^A\otimes\M_{0|3}^B\ket{\psi} = \frac{1}{3}\\ \label{tri}
    \bra{\psi}\M_{2|0}^A\otimes \M_{2|0}^B\ket{\psi} = \bra{\psi}\M_{2|1}^A\otimes \M_{2|0}^B\ket{\psi} = \bra{\psi}\M_{2|0}^A\otimes\M_{2|1}^B\ket{\psi} = \bra{\psi}\M_{2|1}^A\otimes\M_{2|1}^B\ket{\psi} = \frac{1}{3}\\ \label{cetiri}
    \bra{\psi}\M_{0|0}^A\otimes \M_{0|2}^B\ket{\psi} = \bra{\psi}\M_{0|2}^A\otimes \M_{0|2}^B\ket{\psi} = \bra{\psi}\M_{0|0}^A\otimes\M_{0|3}^B\ket{\psi} = \bra{\psi}\M_{0|2}^A\otimes\M_{0|3}^B\ket{\psi} = \frac{1}{3}
\end{align}
Given that in DI scenario we can consider all measurement operators as projectors, eqs.~\eqref{jedan},\eqref{dva} imply that the norms of vectors $\M_{2|0}^A\otimes \idd^B\ket{\psi}$, $\M_{2|1}^A\otimes \idd^B\ket{\psi}$, $\M_{0|0}^A\otimes \idd^B\ket{\psi}$, $\M_{0|2}^A\otimes \idd^B\ket{\psi}$, $\idd^A\otimes\M_{2|0}^B\ket{\psi}$,$\idd^A\otimes\M_{2|1}^B\ket{\psi}$,$\idd^A\otimes\M_{0|2}^B\ket{\psi}$ and $\idd^A\otimes\M_{0|3}^B\ket{\psi}$ is equal to $1/\sqrt{3}$. Eqs.~\eqref{tri},\eqref{cetiri} show that the inner product between two vectors of norm $1/\sqrt{3}$ as determined from eqs.~\eqref{jedan},\eqref{dva} is equal to $1/3$, which by saturation of the  Cauchy–Bunyakovsky–Schwarz inequality implies that vectors figuring in the inner product are parallel:
\begin{align}\label{peta}
    \M_{2|0}^A\otimes \idd^B\ket{\psi} = \idd^A\otimes\M_{2|0}^B\ket{\psi}, \qquad &\M_{2|0}^A\otimes \idd^B\ket{\psi} = \idd^A\otimes\M_{2|1}^B\ket{\psi},\\
    \M_{2|1}^A\otimes \idd^B\ket{\psi} = \idd^A\otimes\M_{2|0}^B\ket{\psi}, \qquad &\M_{2|1}^A\otimes \idd^B\ket{\psi} = \idd^A\otimes\M_{2|1}^B\ket{\psi},\\ \label{sedma}
    \M_{0|0}^A\otimes \idd^B\ket{\psi} = \idd^A\otimes\M_{0|2}^B\ket{\psi}, \qquad &\M_{0|0}^A\otimes \idd^B\ket{\psi} = \idd^A\otimes\M_{0|3}^B\ket{\psi},\\ \label{osma}
    \M_{0|2}^A\otimes \idd^B\ket{\psi} = \idd^A\otimes\M_{0|2}^B\ket{\psi}, \qquad &\M_{0|2}^A\otimes \idd^B\ket{\psi} = \idd^A\otimes\M_{0|3}^B\ket{\psi}
\end{align}
The tensor product form of the state $\ket{\psi} = \ket{\psi}^{AB_1}\otimes\ket{\psi}^{B_2C}$ and eqs.~\eqref{peta}-\eqref{osma} imply that measurements $\M_{2|0}^B$, $\M_{2|1}^B$, $\M_{0|2}^B$ and $\M_{0|3}^B$ act nontrivially only on Hilbert space $\mathcal{H}^{B_1}$.
Given that for all measurements it holds $\sum_{i}\M_{i|j} = \idd$, eqs.~\eqref{peta}-\eqref{osma} imply the following relations:
\begin{align}\label{skoro}
    \ket{\psi_{0,1,Z,A}} = \ket{\psi_{0,1,D,B}} = \ket{\psi_{0,1,E,B}} = \ket{\psi_{0,1,X,A}} \equiv \ket{\psi_{0,1}}\\
     \ket{\psi_{1,2,Z,A}} = \ket{\psi_{1,2,D,B}} = \ket{\psi_{1,2,E,B}} = \ket{\psi_{1,2,X,A}} \equiv \ket{\psi_{1,2}}
\end{align}
and furthermore the norm of $\ket{\psi_{0,1}}$ and $\ket{\psi_{1,2}}$ is equal to $\sqrt{\frac{2}{3}}$. 
Also since $\ket{\psi_{0,1,D,B}} = \ket{\psi_{0,1,E,B}}$ and $\ket{\psi_{1,2,D,B}} = \ket{\psi_{1,2,E,B}}$, given the definition of $\idd_{0,1}^B$ and $\idd_{1,2}^B$, it must be
\begin{equation}
    \ket{\psi_{0,1,B}} = \ket{\psi_{0,1}}, \qquad \ket{\psi_{1,2,B}} = \ket{\psi_{1,2}}.
\end{equation}

The simulation of reference correlations implies:
\begin{equation}
    \bra{\psi}\Z_{0,1}^A\otimes\D_{0,1}^B + \Z_{0,1}^A\otimes\E_{0,1}^B + \X_{0,1}^A\otimes\D_{0,1}^B - \X_{0,1}^A\otimes\E_{0,1}^B\ket{\psi} = \frac{2}{3}2\sqrt{2}
\end{equation}
A variant of sum-of-squares (SOS) decomposition of generalized shifted Bell operator reads:
\begin{multline}\label{sos}
    \sqrt{2}\left[\frac{\idd^A_{0,1,Z} + \idd^A_{0,1,X} + \idd^B_{0,1,D} + \idd^B_{0,1,E}}{\sqrt{2}} - \left(\Z_{0,1}^A\otimes\left(\D_{0,1}^B + \E_{0,1}^B\right)+ \X_{0,1}^A\otimes\left(\D_{0,1}^B-\E_{0,1}^B\right) \right)\right] = \\ = \left(\Z^A_{0,1} - \frac{\D_{0,1}^B + \E_{0,1}^B}{\sqrt{2}}\right)^2 + \left(\X^A_{0,1} - \frac{\D_{0,1}^B - \E_{0,1}^B}{\sqrt{2}}\right)^2,
\end{multline}
where we used the fact that ${\Z^A_{0,1}}^2 =  \idd^A_{0,1,Z}$, ${\X^A_{0,1}}^2 = \idd^A_{0,1,X}$, ${\D^B_{0,1}}^2 =  \idd^B_{0,1,D}$ and ${\E^B_{0,1}}^2 =  \idd^B_{0,1,E}$. Given that
\begin{equation}
\bra{\psi}\idd^A_{0,1,Z} + \idd^A_{0,1,X} + \idd^B_{0,1,D,} + \idd^B_{0,1,E}\ket{\psi} = \frac{8}{3},
\end{equation}
the l.h.s. of eq.~\eqref{sos} is equal to $0$, which means that both sums on the r.h.s. must be equal to $0$ as well, as they are squares of operators implying they must be both nonnegative. Hence, recalling notation introduced in~\eqref{hatoperators} this implies:
\begin{equation}\label{zz}
    \Z^{A}_{0,1}\ket{\psi} = {\hat{\Z}}^{B}_{0,1}\ket{\psi}, \qquad \X^{A}_{0,1}\ket{\psi} = {\hat{\X}}^{B}_{0,1}\ket{\psi}
\end{equation}
Since $\ket{\psi} = \ket{\psi}^{AB_1}\otimes\ket{\psi}^{B_2C}$, and on the l.h.s. of eqs.~\eqref{zz} identity operators act on $\ket{\psi}^{B_2C}$ we can conclude that operators ${\hat{\Z}}^{B}_{0,1}$ and ${\hat{\X}}^{B}_{0,1}$ act nontrivially only on Hilbert space $\mathcal{H}^{B_1}$. 

As operators ${\hat{Z}}_{0,1}^{B}$ and ${\hat{X}}_{0,1}^{B}$ anticommute by construction, operators ${{Z}}_{0,1}^{A}$ and ${{X}}_{0,1}^{A}$ anticommute on the support of $\tr_B[\proj{\psi}]$:
\begin{align} \nonumber
\{{{\X}}_{0,1}^{A},{{\Z}}_{0,1}^{A}\}\ket{\psi} &= {{\X}}_{0,1}^{A}{{\Z}}_{0,1}^{A}\ket{\psi} + {{\Z}}_{0,1}^{A}{{\X}}_{0,1}^{A}\ket{\psi}\\ \nonumber
&= {{\X}}_{0,1}^{A}{\hat{\Z}}_{0,1}^{B}\ket{\psi} + {{\Z}}_{0,1}^{A}{\hat{\X}}_{0,1}^{B}\ket{\psi}\\ \nonumber
&= {\hat{\Z}}_{0,1}^{B}{{\X}}_{0,1}^{A}\ket{\psi} + {\hat{\X}}_{0,1}^{B}{{\Z}}_{0,1}^{A}\ket{\psi}\\ \nonumber
&= {\hat{\Z}}_{0,1}^{B}{\hat{\X}}_{0,1}^{B}\ket{\psi} + {\hat{\X}}_{0,1}^{B}{\hat{\Z}}_{0,1}^{B}\ket{\psi}\\ \label{xz01}
&= 0.
\end{align}
The same procedure can be repeated for the correlations among $\Z^A_{1,2}$, $\X^A_{1,2}$, $\hat{\D}^A_{1,2}$ and $\hat{\E}^A_{1,2}$ to obtain:
\begin{align}\label{zz12}
    \Z^{A}_{1,2}\ket{\psi} = {\hat{\Z}}^{B}_{1,2}\ket{\psi}, &\qquad \X^{A}_{1,2}\ket{\psi} = {\hat{\X}}^{B}_{1,2}\ket{\psi}\\ \label{xz12}
    \{{{\X}}_{1,2}^{A},{{\Z}}_{1,2}^{A}\}\ket{\psi} = 0, &\qquad  \{{\hat{\X}}_{1,2}^{B},{\hat{\Z}}_{1,2}^{B}\}\ket{\psi} = 0
\end{align}
Again, as is the case for ${\hat{\Z}}^{B}_{0,1}$ and ${\hat{\X}}^{B}_{0,1}$, operators ${\hat{\Z}}^{B}_{1,2}$ and ${\hat{\X}}^{B}_{1,2}$ act nontrivially only on the Hilbert space $\mathcal{H}^{B_2}$.  As noted earlier, hatted operators do not necessarilly have eigenvalues $-1$ and $1$, and in prospect of using unitary operators to build self-testing isometry we defined the regularized operators $\Z_{0,1}^B$, $\X_{0,1}^B$, $\Z_{1,2}^B$ and $\X_{1,2}^B$ which are unitary by construction and they act on $\ket{\psi}$ in the same way as ${\hat{\Z}}_{0,1}^B$, ${\hat{\X}}_{0,1}^B$, ${\hat{\Z}}_{1,2}^B$ and ${\hat{\X}}_{1,2}^B$ respectfully. The proof of this is described in details in Appendix A2 in~\cite{SupicBowles}. 
Let us now introduce projective operators
\begin{align}
    \Pp_{0}^B &= \frac{\idd_{0,1,Z}^B + \Z_{0,1}^B}{2},\\
    \Pp_{1}^B &= \frac{\idd_{0,1,Z}^B - \Z_{0,1}^B}{2} = \frac{\idd_{0,1,Z}^B + \Z_{1,2}^B}{2},\\
    \Pp_{2}^B &= \frac{\idd_{0,1,Z}^B - \Z_{1,2}^B}{2}
\end{align}
which together with  $\omega = \exp{\frac{i2\pi}{3}}$ form the operators used to build the self-testing isometry
\begin{align}\label{za}
    \Z^A &= \sum_{j=0}^2\omega^j\M^A_{j|0},\\ \label{zb}
    \Z^B &= \sum_{j=0}^2\omega^j\Pp^B_{j},\\ \label{x1}
    {\X^{(1)}}^{A} &= \X_{0,1}^{A} + \idd^{A} - \idd_{0,1,X}^{A},\qquad  {\X^{(1)}}^{B} = \X_{0,1}^{B} + \idd^{B} - \idd_{0,1}^{B},\\
    {\X^{(2)}}^{A} &= {\X^{(1)}}^{A}\left(\idd^{A} - \idd_{1,2,X}^{A}+\X_{1,2}^{A}\right), \qquad {\X^{(2)}}^{B} = {\X^{(1)}}^{B}\left(\idd^{B} - \idd_{1,2}^{B}+\X_{1,2}^{B}\right). \label{x2}
\end{align}
Operators $\Z^{A/B}$, ${\X^{(1)}}^{A/B}$
and ${\X^{(2)}}^{A/B}$ are unitary. Indeed, operators $\Z^{A/B}$ are constructed as a sum of projectors with eigenvalues squaring to $1$. For operators ${\X^{(1)}}^{A}$
and ${\X^{(2)}}^{A}$ unitarity can be proven in the following way
\begin{align}
    {\X^{(1)}}^{A}{{\X^{(1)}}^\dagger}^{A} &=\left(\X_{0,1}^{A} + \idd^{A} - \idd_{0,1,X}^{A}\right)\left(\X_{0,1}^{A} + \idd^{A} - \idd_{0,1,X}^{A}\right)\\
    &= \idd\\
     {\X^{(2)}}^{A}{{\X^{(2)}}^\dagger}^{A} &= {\X^{(1)}}^{A}\left(\idd^{A} - \idd_{1,2,X}^{A}+\X_{1,2}^{A}\right)\left(\idd^{A} - \idd_{1,2,X}^{A}+{\X_{1,2}^\dagger}^{A}\right){{\X^{(1)}}^\dagger}^{A}\\
     &= {\X^{(1)}}^{A}{{\X^{(1)}}^\dagger}^{A}\\
     &= \idd.
\end{align}
The same unitarity proof holds for ${\X^{(1)}}^{B}$
and ${\X^{(2)}}^{B}$. Note that ${\X^{(1)}}^{A}$ can alternatively be written as
\begin{align}\label{x1Aalt}
    {\X^{(1)}}^{A} = {\X^{A}_{0,1}} + \M^{A}_{2|0},
\end{align}
while for ${\X^{(1)}}^{B}$ holds:
\begin{equation}\label{x1Balt}
    {\X^{(1)}}^{B}\ket{\psi} = {\X^{B}_{0,1}}\ket{\psi} + \M^{B}_{2|0}\ket{\psi}.
\end{equation}
Similarly we obtain the following relations:
\begin{align}\label{X12alt}
    \idd^{A} - \idd_{1,2,X}^{A}+\X_{1,2}^{A} &= \M_{0|2} + \X_{1,2}^{A},\\
    \left(\idd^{B} - \idd_{1,2}^{B}+\X_{1,2}^{B}\right)\ket{\psi} &= \left(\M_{0|2}^B + \X_{1,2}^{B}\right)\ket{\psi}.
\end{align}
Eqs.~\eqref{zz} and \eqref{zz12} imply
\begin{align}\label{parallelz}
    \Z^A\ket{\psi} &= \Z^B\ket{\psi},\\
    {\X^{(1)}}^{A}\ket{\psi} &= {\X^{(1)}}^{B}\ket{\psi}\\
    {\X^{(2)}}^{A}\ket{\psi} &= {{\X^{(2)}}^\dagger}^{B}\ket{\psi}.
\end{align}
Eq.~\eqref{parallelz} implies
\begin{equation}\label{deltaz}
    \M_{j|0}^A\otimes\Pp_k^B\ket{\psi} = \delta_{j,k}\M_{j|0}^A\ket{\psi}.
\end{equation}

\begin{figure}
	\begin{center}
	\includegraphics[width=0.6\columnwidth]{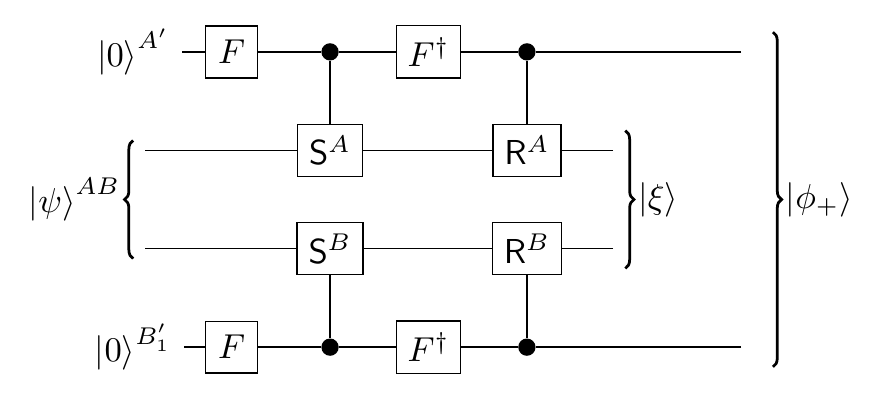}
	\end{center}
	\caption{Isometry for self-testing  maximally entangled pair of qutrits.
	}
	\label{fig:iso}
\end{figure}
The self-testing isometry $\Phi = \Phi_A\otimes\Phi_{B_1}$ is given on Fig.~\ref{fig:iso}. Note that operators $\Z^B$ and ${\X^{(1/2)}}^B$ are built from measurements that act nontrivially only on Hilbert space $\mathcal{H}^{B_1}$, which means that the same holds for the isometry $\Phi_{B_1}$. Operator $F$ is the Fourier transform acting in the following way: $F\ket{j} = \frac{1}{\sqrt{3}}\sum_{k=0}^2\omega^{jk}\ket{k}$. The controlled gates $C\Ss^{A/B}$ and $C\R^{A/B}$ are defined as:
\begin{align}
  C\Ss^{A/B}\ket{j}^{A'/B_1'}\ket{\psi} &= \ket{j}^{A'/B_1'}{\Z^j}^{A/B}\ket{\psi},\\
  C\R^{A/B}\ket{j}^{A'/B_1'}\ket{\psi} &= \ket{j}^{A'/B_1'}{\X^{(j)}}^{A/B}\ket{\psi},
\end{align}
where $\Z^j$ is simply $j$-th power of $\Z$, $\X^{(0)} = \idd$ and $\X^{(j)}$ are given in eqs.~\eqref{x1},\eqref{x2}. The output state of the isometry is:
\begin{align}
    \Phi(\ket{\psi}^{AB}\otimes\ket{00}^{A'B_1'}) = \sum_{j,k=0}^2{\X^{(j)}}^A\M_{j|0}^A\otimes{\X^{(k)}}^B\Pp_k^B\ket{\psi}^{ABC}\otimes\ket{jk}^{A'B_1'},
\end{align}
where we used the identities
\begin{align}
    \M_{0|0}^A = \frac{1}{3}\left({\Z^0}^A + {\Z^1}^A+ {\Z^2}^A\right), \quad &\M_{1|0}^A = \frac{1}{3}\left({\Z^0}^A + \omega{\Z^1}^A+ \omega^*{\Z^2}^A\right),\quad \M_{2|0}^A = \frac{1}{3}\left({\Z^0}^A + \omega^*{\Z^1}^A+ \omega{\Z^2}^A\right),\\
    \Pp_{0}^B = \frac{1}{3}\left({\Z^0}^B + {\Z^1}^B+ {\Z^2}^B\right), \quad &\Pp_{1}^B = \frac{1}{3}\left({\Z^0}^B + \omega{\Z^1}^B+ \omega^*{\Z^2}^B\right),\quad \Pp_{2}^B = \frac{1}{3}\left({\Z^0}^B + \omega^*{\Z^1}^B+ \omega{\Z^2}^B\right),
\end{align}
as per definition of $\Z^{A/B}$ given in eqs.~\eqref{za},\eqref{zb}.
Eqs.~\eqref{deltaz} imply that only surviving elements of the sum are those containing $\ket{jk}^{A'B'}$ such that $j=k$. Hence, the explicitly written full output state $\Phi(\ket{\psi}^{AB}\otimes\ket{00}^{A'B_1'})$ reads
\begin{align}
    \Phi(\ket{\psi}^{ABC}\otimes\ket{00}^{A'B_1'}) = \M_{0|0}^A\ket{\psi}^{ABC}\otimes\ket{00}^{A'B_1'} + {\X^{(1)}}^A\M_{1|0}^A\otimes{\X^{(1)}}^B\ket{\psi}^{ABC}\otimes\ket{11}^{A'B_1'} + {\X^{(2)}}^A\M_{2|0}^A\otimes{\X^{(2)}}^B\ket{\psi}^{ABC}\otimes\ket{22}^{A'B_1'}
\end{align}
We simplify the second term on the r.h.s. of the last eq.:
\begin{align}
    {\X^{(1)}}^A\M_{1|0}^A\otimes{\X^{(1)}}^B\ket{\psi} &= \left({\X^{A}_{0,1}} + \M^{A}_{2|0}\right)\frac{\idd_{0,1,Z}-\Z_{0,1}^A}{2}\otimes\left({\X^{B}_{0,1}} + \M^{B}_{2|0}\right)\ket{\psi}\\&=
    \frac{\idd_{0,1,Z}+\Z_{0,1}^A}{2}\left({\X^{A}_{0,1}} + \M^{A}_{2|0}\right)\otimes\left({\X^{B}_{0,1}} + \M^{B}_{2|0}\right)\ket{\psi}\\
    &= \M_{0|0}^A\ket{\psi}
\end{align}
In the first line we used eqs.~\eqref{x1Aalt} and \eqref{x1Balt}. To get the second line we used the anticommutation between $\X_{0,1}^A$ and $\Z_{0,1}^A$. The third eq. was obtained from~\eqref{zz} and the fact that ${X_{0,1}^A}^2\ket{\psi} = \idd_{0,1,X}^A\ket{\psi} = \ket{\psi_{0,1}}$. In the last equation we used the fact that $\M_{0|0}^A =  \frac{\idd_{0,1,Z}+\Z_{0,1}^A}{2}$ and  $\ket{\psi_{0,1}} = (\M_{0|0}^A+\M_{1|0}^A)\ket{\psi}$.
Hence, we get simplified output state 
\begin{align}
    \Phi(\ket{\psi}^{ABC}\otimes\ket{00}^{A'B_1'}) = \M_{0|0}^A\ket{\psi}^{ABC}\otimes\left(\ket{00}^{A'B'} + \ket{11}^{A'B_1'}\right) + {\X^{(2)}}^A\M_{2|0}^A\otimes{\X^{(2)}}^B\ket{\psi}^{ABC}\otimes\ket{22}^{A'B_1'}.
\end{align}
Now we take care of the last term
\begin{align}
    {\X^{(2)}}^A\M_{2|0}^A\otimes{\X^{(2)}}^B\ket{\psi} &= \left({\X^{A}_{0,1}} + \M^{A}_{2|0}\right)\left(\M_{0|2}^{A}+\X_{1,2}^{A}\right)\frac{\idd_{1,2,Z}-\Z^A_{1,2}}{2}\otimes\left({\X^{B}_{0,1}} + \M^{B}_{2|0}\right)\left(\M_{0|2}^{B}+\X_{1,2}^{B}\right)\ket{\psi}\\
    &= \left({\X^{A}_{0,1}} + \M^{A}_{2|0}\right)\frac{\idd_{1,2,Z}+\Z^A_{1,2}}{2}\left(\M_{0|2}^{A}+\X_{1,2}^{A}\right)\otimes\left({\X^{B}_{0,1}} + \M^{B}_{2|0}\right)\left(\M_{0|2}^{B}+\X_{1,2}^{B}\right)\ket{\psi}\\
    &= \left({\X^{A}_{0,1}} + \M^{A}_{2|0}\right)\frac{\idd_{0,1,Z}-\Z^A_{0,1}}{2}\otimes\left({\X^{B}_{0,1}} + \M^{B}_{2|0}\right)\ket{\psi}\\
    &= \frac{\idd_{0,1,Z}+\Z^A_{0,1}}{2}\left({\X^{A}_{0,1}} + \M^{A}_{2|0}\right)\otimes\left({\X^{B}_{0,1}} + \M^{B}_{2|0}\right)\ket{\psi}\\
    &= \M_{0|0}^A\ket{\psi}.
\end{align}
In the first line we used eqs.~\eqref{X12alt}. In the second line we used the anticommutation relation between $\Z_{1,2}^A$ and $\X_{1,2}^A$ (cf.~\eqref{xz12}). To get the third line we used the fact that $\left(\M_{0|2}^{A}+\X_{1,2}^{A}\right)\otimes \left(\M_{0|2}^{B}+\X_{1,2}^{B}\right)\ket{\psi} = \ket{\psi}$ and also equality $\idd_{1,2,Z}+\Z^A_{1,2} = \idd_{0,1,Z}-\Z^A_{0,1}$. The last two lines are the consequence of the anticommutation relation between $\Z_{1,2}^A$ and $\X_{1,2}^A$ (cf.~\eqref{xz01}) and relation $\left(\M_{2|0}^{A}+\X_{0,1}^{A}\right)\otimes \left(\M_{2|0}^{B}+\X_{0,1}^{B}\right)\ket{\psi} = \ket{\psi}$. Finally, we obtain the self-testing remark for the state:
\begin{align}\label{ststate}
    \Phi(\ket{\psi}^{ABC}\otimes\ket{00}^{A'B_1'}) = \sqrt{3}\M_{0|0}^A\ket{\psi}^{ABC}\otimes\ket{\phi_+}^{A'B_1'} \equiv \ket{\xi}^{ABC} \otimes \ket{\phi_+}^{A'B_1'},
\end{align}
where we introduced notation $\ket{\xi}^{ABC} = \left(\sqrt{3}\M_{0|0}^A\ket{\psi}^{AB_1}\right)\otimes\ket{\psi}^{B_2C}$, which is a valid quantum state, because the norm of $\M_{0|0}^A\ket{\psi}^{AB_1}$ is equal to $1/\sqrt{3}$ (cf.~\eqref{sedma}). 
Now we move to the self-testing of measurements. Let us start with measurement $\M_{l|0}^A$, and see how the self-testing isometry maps the state $\M_{l|0}^A\ket{\psi}^{ABC}$:
\begin{align}\nonumber
    \Phi(\M_{l|0}^A\ket{\psi}^{ABC}\otimes\ket{00}^{A'B_1'}) &= \sum_{j,k=0}^2{\X^{(j)}}^A\M_{j|0}^A\M_{l|0}^A\otimes{\X^{(k)}}^B\Pp_k^B\ket{\psi}^{ABC}\otimes\ket{jk}^{A'B_1'}\\ \nonumber
    &= \sum_{j,k=0}^2{\X^{(j)}}^A\delta_{j,l}\M_{l|0}^A\otimes{\X^{(k)}}^B\Pp_k^B\ket{\psi}^{ABC}\otimes\ket{jk}^{A'B_1'}\\ \nonumber
    &= \ket{\xi}^{ABC}\otimes\frac{1}{\sqrt{3}}\ket{ll}^{A'B_1'}\\ \label{stmeas0}
    &= \ket{\xi}^{ABC}\otimes{\M'}_{l|0}^{A'}\ket{\phi_+}^{A'B_1'},
\end{align}
which is exactly the self-testing statement for $\M'_{l|0}$. To get the second line we used orthogonality of projectors $\{\M_{j|0}\}_j$, and the following two lines just reproduce the proof of self-testing the state. Exactly the same proof holds for measurement operators $\M_{2|1}^A$ and $\M_{0|2}^{A}$, as they act on $\ket{\psi}$ in the same way as $\M_{2|0}^A$ and $\M_{0|0}^A$ respectively. Concerning self-testing of $\M^A_{0|1}$ and $\M^A_{1|1}$ given that $\X_{0,1}^A = \M^A_{0|1} - \M^A_{1|1}$  and $\idd_{0,1,X}^A = \M^A_{0|1} + \M^A_{1|1}$, self-testing $\X_{0,1}^A$ and $\idd_{0,1,X}^A$ is equivalent to self-testing $\M^A_{0|1}$ and $\M^A_{1|1}$. Given that $\idd_{0,1,X}^A\ket{\psi} = \idd_{0,1,Z}^A\ket{\psi}$ and $\idd_{0,1,Z} = \M_{0|0}^A + \M_{1|0}^A$ the self-testing statement for $\{\M_{j|0}^A\}_j$ allows to conclude 
\begin{align}\label{idd01st}
    \Phi(\idd_{0,1,X}^A\ket{\psi}^{ABC}\otimes\ket{00}^{A'B_1'}) &= \ket{\xi}^{ABC}\otimes (\proj{0}^{A'} + \proj{1}^{A'})\ket{\phi_+}^{A'B_1'}.
\end{align}
We now turn to self-testing od $\X_{0,1}^A$:
\begin{align}
    \Phi(\X_{0,1}^A\ket{\psi}^{ABC}\otimes\ket{00}^{A'B_1'}) &= \sum_{j,k=0}^2{\X^{(j)}}^A\M_{j|0}^A\X_{0,1}^A\otimes{\X^{(k)}}^B\Pp_k^B\ket{\psi}^{ABC}\otimes\ket{jk}^{A'B_1'}\\
    &= \sum_{j,k=0}^1{\X^{(j)}}^A\M_{j|0}^A\X_{0,1}^A\otimes{\X^{(k)}}^B\Pp_k^B\ket{\psi}^{ABC}\otimes\ket{jk}^{A'B_1'}. \\ 
\end{align}
Note that in the second line the sum has different upper bounds. Let us first prove that all elements of the sum corresponding to $k=2$ vanish:
\begin{align}
\sum_{j=0}^2{\X^{(j)}}^A\M_{j|0}^A\X_{0,1}^A\otimes{\X^{(2)}}^B\Pp_2^B\ket{\psi}^{ABC}\otimes\ket{j2}^{A'B_1'} &= \sum_{j}^2{\X^{(j)}}^A\M_{j|0}^A\X_{0,1}^A\M_{2|0}^A\otimes{\X^{(2)}}^B\ket{\psi}^{ABC}\otimes\ket{j2}^{A'B_1'}\\
&=\sum_{j}^2{\X^{(j)}}^A\M_{j|0}^A\left(\M_{0|1}^A-\M_{1|1}^A\right)\M_{2|1}^A\otimes{\X^{(2)}}^B\ket{\psi}^{ABC}\otimes\ket{j2}^{A'B_1'}\\
&= 0,
\end{align}
where in the second line we used eq.~\eqref{deltaz}, and in the second we used eq.~\eqref{skoro} and definition of $\X_{0,1}^A$. Finally, to get the last line we used orthogonality of projectors corresponding to the same measurement. Now we consider elements of the sum corresponding to $j=2$:
\begin{align}
   \sum_{k=0}^1{\X^{(2)}}^A\M_{2|0}^A\X_{0,1}^A\otimes{\X^{(k)}}^B\Pp_k^B\ket{\psi}^{ABC}\otimes\ket{2k}^{A'B_1'}
  &= \sum_{k=0}^1{\X^{(2)}}^A\M_{2|0}^A\otimes{\X^{(k)}}^B\Pp_k^B\X_{0,1}^B\ket{\psi}^{ABC}\otimes\ket{2k}^{A'B_1'}\\
  &= \sum_{k=0}^1{\X^{(2)}}^A\otimes{\X^{(k)}}^B\Pp_k^B\X_{0,1}^B\Pp_2^B\ket{\psi}^{ABC}\otimes\ket{2k}^{A'B_1'}\\
  &= 0.
\end{align}
In the first line we used eq.~\eqref{zz} and the fact that hatted and nonhatted operators of Bob act on $\ket{\psi}$ in the same way. In the second eq. we used eq.~\eqref{parallelz}, and to get the last eq. we used the fact that ranges of $\Pp_2^B$ and $\X_{0,1}^B$ are orthogonal.
 Let us write the whole remaining state
\begin{align}\nonumber
    \Phi(\X_{0,1}^A\ket{\psi}^{ABC}\otimes\ket{00}^{A'B_1'}) &= \M_{0|0}^A\X_{0,1}^A\otimes\Pp_0^B\ket{\psi}^{ABC}\otimes\ket{00}^{A'B_1'} + \M_{0|0}^A\X_{0,1}^A\otimes{\X^{(1)}}^B\Pp_1^B\ket{\psi}^{ABC}\otimes\ket{01}^{A'B_1'} + \\ \nonumber &\qquad + {\X^{(1)}}^A\M_{1|0}^A\X_{0,1}^A\otimes\Pp_0^B\ket{\psi}^{ABC}\otimes\ket{10}^{A'B_1'} + {\X^{(1)}}^A\M_{1|0}^A\X_{0,1}^A\otimes{\X^{(1)}}^B\Pp_1^B\ket{\psi}^{ABC}\otimes\ket{11}^{A'B_1'}\\ \nonumber &= \X_{0,1}^A\M_{1|0}^A\otimes\Pp_0^B\ket{\psi}^{ABC}\otimes\ket{00}^{A'B_1'} + \X_{0,1}^A\M_{1|0}^A\otimes{\X^{(1)}}^B\Pp_1^B\ket{\psi}^{ABC}\otimes\ket{01}^{A'B_1'} + \\ \nonumber &\qquad + {\X^{(1)}}^A\X_{0,1}^A\M_{0|0}^A\otimes\Pp_0^B\ket{\psi}^{ABC}\otimes\ket{10}^{A'B_1'} + {\X^{(1)}}^A\X_{0,1}^A\M_{0|0}^A\otimes{\X^{(1)}}^B\Pp_1^B\ket{\psi}^{ABC}\otimes\ket{11}^{A'B_1'}\\ \nonumber
    &= \X_{0,1}^A\M_{1|0}^A\otimes{\X^{(1)}}^B\Pp_1^B\ket{\psi}^{ABC}\otimes\ket{01}^{A'B_1'} + {\X^{(1)}}^A\X_{0,1}^A\M_{0|0}^A\otimes\Pp_0^B\ket{\psi}^{ABC}\otimes\ket{10}^{A'B_1'}\\ \nonumber
    &=\M_{0|0}^A\X_{0,1}^A\otimes{\X^{(1)}}^B\ket{\psi}^{ABC}\otimes\ket{01}^{A'B_1'} + \M_{0|0}^A
    \idd_{0,1,X}^A\otimes\Pp_0^B\ket{\psi}^{ABC}\otimes\ket{10}^{A'B_1'}\\ \nonumber 
    &= \M_{0|0}^A\ket{\psi}^{ABC}\otimes (\ket{01}^{A'B_1'} + \ket{10}^{A'B_1'})\\ \label{stX01}
    &\junk^{ABC}\otimes \left(\ketbra{0}{1}^{A'} + \ketbra{1}{0}^{A'}\right)\ket{\phi}^{A'B_1'}.
\end{align}
The second equality (the third and fourth lines) is obtained by using anticommutation relation between $X_{0,1}^A$ and $\Z_{0,1}^A$ (cf.~\eqref{xz01}). In the fifth line we used the fact that $\M_{1|0}^A\otimes\Pp_0^B\ket{\psi} = 0$ and $\M_{0|0}^A\otimes\Pp_1^B\ket{\psi} = 0$ (cf.~\eqref{deltaz}). To get the sixth and seventh lines we used again the anticommutation between $\Z_{0,1}^A$ and $\X_{0,1}^A$ and relation $\X_{0,1}^A{\X^{1}}^A\ket{\psi} = \idd_{0,1,X}\ket{\psi}$.
Given the definition of $\M'_{0|1}$ and $\M'_{1|1}$, eqs.~\eqref{idd01st} and~\eqref{stX01} imply
\begin{align}\label{stmeas1}
    \Phi(\M_{j|1}^A\ket{\psi}^{ABC}\otimes\ket{00}^{A'B_1'}) &= \ket{\xi}^{ABC}\otimes{\M'}_{j|1}^{A'}\ket{\phi_+}^{A'B_1'},
\end{align}
for $j = 0,1,2$. Completely analogous proof holds for self-testing the third Alice's measurement:
\begin{align}\label{stmeas2}
    \Phi(\M_{j|2}^A\ket{\psi}^{ABC}\otimes\ket{00}^{A'B_1'}) &= \ket{\xi}^{ABC}\otimes{\M'}_{j|2}^{A'}\ket{\phi_+}^{A'B_1'},
\end{align}
for $j = 0,1,2$.

Eqs.~\eqref{stmeas0},\eqref{stmeas1} and~\eqref{stmeas2} together imply:
\begin{align}\label{stmeasAlice}
    \Phi(\M_{a|x}^A\ket{\psi}^{ABC}\otimes\ket{00}^{A'B_1'}) &= \ket{\xi}^{ABC}\otimes{\M'}_{a|x}^{A'}\ket{\phi_+}^{A'B_1'},
\end{align}
The set of eqs.~\eqref{peta}-\eqref{osma} implies that operators $\M_{2|0}^B$, $\M_{2|1}^B$, $\M_{0|2}^B$ and $\M_{0|3}^B$ act on $\ket{\psi}$ in the same way as Alice's measurements self-tested in~\eqref{stmeasAlice} which implies:
\begin{align}\label{jedanmeasB}
   \Phi(\M_{2|0}^B\ket{\psi}^{ABC}\otimes\ket{00}^{A'B_1'}) &= \ket{\xi}^{ABC}\otimes{\M'}_{2|0}^{B_1'}\ket{\phi_+}^{A'B_1'}, \\
   \Phi(\M_{2|1}^B\ket{\psi}^{ABC}\otimes\ket{00}^{A'B_1'}) &= \ket{\xi}^{ABC}\otimes{\M'}_{2|1}^{B_1'}\ket{\phi_+}^{A'B_1'}, \\
   \Phi(\M_{0|2}^B\ket{\psi}^{ABC}\otimes\ket{00}^{A'B_1'}) &= \ket{\xi}^{ABC}\otimes{\M'}_{0|2}^{B_1'}\ket{\phi_+}^{A'B_1'}, \\ \label{cetirimeasB}
   \Phi(\M_{0|3}^B\ket{\psi}^{ABC}\otimes\ket{00}^{A'B_1'}) &= \ket{\xi}^{ABC}\otimes{\M'}_{0|3}^{B_1'}\ket{\phi_+}^{A'B_1'}.
\end{align}
Let us now define reference operators:
\begin{align}
\X'_{0,1} = \begin{bmatrix}
    0 & 1 & 0\\ 1 & 0 & 0\\ 0& 0 & 0
    \end{bmatrix}, \qquad \Z'_{0,1} = \begin{bmatrix}
    1 & 0 & 0\\ 0 & -1 & 0\\ 0& 0 & 0
    \end{bmatrix}, \qquad \X'_{1,2} = \begin{bmatrix}
    0 & 0 & 0\\ 0 & 0 & 1\\ 0& 1 & 0
    \end{bmatrix}, \qquad \Z'_{1,2} = \begin{bmatrix}
    0 & 0 & 0\\ 0 & 1 & 0\\ 0& 0 & -1
    \end{bmatrix}.
\end{align}
Eqs.~\eqref{stmeasAlice} implies:
\begin{align}\label{jedanOp}
    \Phi(\X_{0,1}^A\ket{\psi}^{ABC}\otimes\ket{00}^{A'B_1'}) &= \ket{\xi}^{ABC}\otimes{\X'}_{0,1}^{A'}\ket{\phi_+}^{A'B_1'},\\
    \Phi(\Z_{0,1}^A\ket{\psi}^{ABC}\otimes\ket{00}^{A'B_1'}) &= \ket{\xi}^{ABC}\otimes{\Z'}_{0,1}^{A'}\ket{\phi_+}^{A'B_1'},\\
    \Phi(\X_{1,2}^A\ket{\psi}^{ABC}\otimes\ket{00}^{A'B_1'}) &= \ket{\xi}^{ABC}\otimes{\X'}_{1,2}^{A'}\ket{\phi_+}^{A'B_1'},\\ \label{cetiriOp}
    \Phi(\Z_{1,2}^A\ket{\psi}^{ABC}\otimes\ket{00}^{A'B_1'}) &= \ket{\xi}^{ABC}\otimes{\Z'}_{1,2}^{A'}\ket{\phi_+}^{A'B_1'}.
\end{align}
Now eqs.~\eqref{zz} and \eqref{zz12} together with the set of eqs.~\eqref{jedanOp}-\eqref{cetiriOp} lead to:
\begin{align}\label{jedanOpB}
    \Phi(\X_{0,1}^B\ket{\psi}^{ABC}\otimes\ket{00}^{A'B_1'}) &= \ket{\xi}^{ABC}\otimes{\X'}_{0,1}^{B_1'}\ket{\phi_+}^{A'B_1'},\\
    \Phi(\Z_{0,1}^B\ket{\psi}^{ABC}\otimes\ket{00}^{A'B_1'}) &= \ket{\xi}^{ABC}\otimes{\Z'}_{0,1}^{B_1'}\ket{\phi_+}^{A'B_1'},\\
    \Phi(\X_{1,2}^B\ket{\psi}^{ABC}\otimes\ket{00}^{A'B_1'}) &= \ket{\xi}^{ABC}\otimes{\X'}_{1,2}^{B_1'}\ket{\phi_+}^{A'B_1'},\\ \label{cetiriOpB}
    \Phi(\Z_{1,2}^B\ket{\psi}^{ABC}\otimes\ket{00}^{A'B_1'}) &= \ket{\xi}^{ABC}\otimes{\Z'}_{1,2}^{B_1'}\ket{\phi_+}^{A'B_1'},
\end{align}
and equivalently
\begin{align}\label{jedanOpD}
    \Phi(\D_{0,1}^B\ket{\psi}^{ABC}\otimes\ket{00}^{A'B_1'}) &= \ket{\xi}^{ABC}\otimes\frac{{\X'}_{0,1}^{B_1'}+{\Z'}_{0,1}^{B_1'}}{\sqrt{2}}\ket{\phi_+}^{A'B_1'},\\
\Phi(\E_{0,1}^B\ket{\psi}^{ABC}\otimes\ket{00}^{A'B_1'}) &= \ket{\xi}^{ABC}\otimes\frac{{\Z'}_{0,1}^{B_1'}-{\X'}_{0,1}^{B_1'}}{\sqrt{2}}\ket{\phi_+}^{A'B_1'},\\
 \Phi(\D_{1,2}^B\ket{\psi}^{ABC}\otimes\ket{00}^{A'B_1'}) &= \ket{\xi}^{ABC}\otimes\frac{{\X'}_{1,2}^{B_1'}+{\Z'}_{1,2}^{B_1'}}{\sqrt{2}}\ket{\phi_+}^{A'B_1'},\\ \label{cetiriOpE}
   \Phi(\E_{1,2}^B\ket{\psi}^{ABC}\otimes\ket{00}^{A'B_1'}) &= \ket{\xi}^{ABC}\otimes\frac{{\Z'}_{1,2}^{B_1'}-{\X'}_{1,2}^{B_1'}}{\sqrt{2}}\ket{\phi_+}^{A'B_1'}.
\end{align}
The last set of equations together with definitions of $\D_{j,k}$ and $\E_{j,k}$ and eqs. \eqref{jedanmeasB}-\eqref{cetirimeasB} imply:
\begin{align}\label{stmeasBob}
    \Phi(\M_{b_1|y}^B\ket{\psi}^{ABC}\otimes\ket{00}^{A'B_1'}) &= \ket{\xi}^{ABC}\otimes{\M'}_{b_1|y}^{B_1'}\ket{\phi_+}^{A'B_1'}.
\end{align}
Eqs.~\eqref{stmeasAlice} and~\eqref{stmeasBob} together give
\begin{align}\label{stmeasAB}
    \Phi(\M_{a|x}^A\otimes\M_{b_1|y}^B\ket{\psi}^{ABC}\otimes\ket{00}^{A'B_1'}) &= \ket{\xi}^{ABC}\otimes{\M'}_{a|x}^{A'}\otimes{\M'}_{b_1|y}^{B_1'}\ket{\phi_+}^{A'B_1'},
\end{align}
where
\begin{align}
    \ket{\xi}^{ABC} &= \left(\sqrt{3}\M_{0|0}^A\ket{\psi}^{AB_1}\right)\otimes\ket{\psi}^{B_2C}\\
    &\equiv \ket{\xi_1}^{AB_1}\otimes\ket{\psi}^{B_2C}.
\end{align}

\section{Proof of Eq. (6)}\label{suppmat2}

Eq.~\eqref{ststate} together with~\eqref{stmeasBob} leads to
\begin{equation}
\Phi_{B_1}^\dagger\left(\M^B_{b_1|y}\otimes\idd^{B_1'}\right) = \idd^{B}\otimes{\M'}_{b_1|y_1}^{B'_1}.
\end{equation}
Since $\M'_{b_1|y_1}$ are projective it must be $\Phi_{B_1}^\dagger\left(\M^B_{b_1,b_2|y}\otimes\idd^{B_1'}\right) = \K_{b_1,b_2|y}^{B}\otimes{\M'}_{b_1|y}^{B'_1}$, where $\K_{b_1,b_2|y}$ is positive semidefinite and $\sum_{b_2}\K_{b_1,b_2|y} = \idd$. With this insight, and eqs.~\eqref{stmeasAlice} and \eqref{stmeasBob} for every collection of inputs $x,y,z$ and outputs $a,b,c$  the equivalence between reference correlations given in eq.~\eqref{ReferenceCorrelations} and physical correlations given in~\eqref{PhysicalCorrelations} implies:
\begin{align}
   \tr\left[\left({\M'}_{a|x}^{A'}\otimes {\M'}_{b_1|y}^{B'_1}\right){\phi}_+^{A'B'_1}\right]\tr\left[\left({\M'}_{b_2|y}^{B'_2}\otimes{\M'}_{c|z}^{C'}\right){\phi}_+^{B'_2C'}\right] = \bra{\psi}^{AB_1}\otimes\bra{\psi}^{B_2C}{\M}_{a|x}^{A}\otimes {\M}_{b_1,b_2|y}^{B_1B_2}\otimes\M_{c|z}^C\ket{\psi}^{AB_1}\otimes\ket{\psi}^{B_2C},\\
   = \bra{\xi}^{ABC}\idd^{A}\otimes\K_{b_1,b_2|y}^{B}\otimes \M_{c|z}^C\ket{\xi}^{ABC} \bra{\phi_+}^{A'B_1'}{\M'}_{a|x}^{A'}\otimes {\M'}_{b_1|y}^{B'_1}\ket{\phi_+}^{A'B_1'}
\end{align}
By cancelling the same terms on two sides of equality we obtain:
\begin{equation}
    \tr\left[\left({\M'}_{b_2|y}^{B'_2}\otimes{\M'}^{C'}_{c|z}\right){\phi}_+^{B'_2C'}\right] = \bra{\xi}^{ABC}\idd^{A}\otimes\K_{b_1,b_2|y}^{B}\otimes \M_{c|z}^C\ket{\xi}^{ABC}
\end{equation}
Since for every fixed $b_1$ the set of operators $\{\K_{b_1,b_2|y}\}_{b_2}$ represents a valid measurement, together with Charlie's measurements and state $\ket{\xi}$ we have a physical experiment which satisfies conditions given in Lemma~\ref{lem2}.  Hence, we can repeat exactly the same procedure as in Appendix~\ref{suppmat1} to establish existence ofthe $\Phi'= \Phi_{B_2}\otimes\Phi_C$ such that $\Phi_{B_2}$ acts nontrivially only on Hilbert space $\mathcal{H}^{B_2}$ and the whole isometry transforms state $\ket{\xi}^{ABC}$ as follows
\begin{equation}\label{isoprime}
    \Phi'\left(\K_{b_1,b_2|y}\otimes\M_{c|z}\ket{\xi}^{ABC}\otimes\ket{00}^{B_2'C'}\right) =  \ket{\tilde{\xi}}^{ABC}\otimes\left(\M'_{b_2|y}\otimes\M'_{c|z}\ket{\phi_+}^{B_2'C'}\right),
\end{equation}
which is the analogue of~\eqref{stmeasAB} and where $\ket{\tilde{\xi}}$ is defined as
\begin{align}
    \ket{\tilde{\xi}}^{ABC} &= \ket{\xi_1}^{AB_1}\otimes\left(\sqrt{3}\M_{0|0}^C\ket{\psi}^{B_2C}\right)\\
    &\equiv \ket{\xi_1}^{AB_1}\otimes\ket{\xi_2}^{B_2C}
\end{align}
Let us now rewrite eq.~\eqref{stmeasAB}
\begin{align}
    \Phi\left(\M_{a|x}^A\otimes\M_{b_1|y}^B\otimes\M_{c|z}^C\ket{\psi}^{AB_1}\otimes\ket{\psi}^{B_2C}\otimes\ket{00}^{A'B_1'}\right) &= \M_{c|z}^C\ket{\xi}^{ABC}\otimes{\M'}_{a|x}^{A'}\otimes{\M'}_{b_1|y_1}^{B_1'}\ket{\phi_+}^{A'B_1'}.
\end{align}
Insights from the beginning of this Appendix allow us to write further
\begin{align}
 \Phi\left(\M_{a|x}^A\otimes\M_{b_1,b_2|y}^B\otimes\M_{c|z}^C\ket{\psi}^{AB_1}\otimes\ket{\psi}^{B_2C}\otimes\ket{00}^{A'B_1'}\right) &= \K_{b_1,b_2|y}^B\otimes\M_{c|z}^C\ket{\xi}^{ABC}\otimes{\M'}_{a|x}^{A'}\otimes{\M'}_{b_1|y}^{B_1'}\ket{\phi_+}^{A'B_1'}.   
\end{align}
By acting on the r.h.s. of this equation with $\Phi'$, use eq.~\eqref{isoprime}, taking into account that $\Phi_{B_1}$ acts nontrivially only on $\mathcal{H}^{B_1}$, while $\Phi_{B_2}$ acts nontrivially only on $\mathcal{H}^{B_2}$, and summing over $b_1$ and $b_2$ we obtain:
\begin{equation}\label{finalst}
    \Phi_A\otimes\Phi_{B_1}\circ\Phi_{B_2}\otimes\Phi_C\left(\M_{a|x}\ket{\psi}^{A_1B_1}\otimes\M_{c|z}\ket{\psi}^{B_2C}\otimes\ket{0000}^{A'B_1'B_2'C'}\right) =  \ket{\tilde{\xi}}^{AB_1B_2C}\otimes\M'_{a|x}\ket{\phi_+}^{A'B_1'}\otimes\M'_{c|z}\ket{\phi_+}^{B_2'C'},
\end{equation}
which completes the proof of eq.~\eqref{iso1} if we denote $\tilde{\Phi} = \Phi_A\otimes\Phi_{B_1}\circ\Phi_{B_2}\otimes\Phi_C$.

\bibliography{bibliography}

\end{document}